\documentclass[conference,a4paper]{IEEEtran}
\makeatletter
\def\ps@headings{%
\def\@oddhead{\mbox{}\scriptsize\rightmark \hfil \thepage}%
\def\@evenhead{\scriptsize\thepage \hfil \leftmark\mbox{}}%
\def\@oddfoot{}%
\def\@evenfoot{}}
\makeatother
\usepackage{simplemargins}
\settopmargin{0.92in}
\setbottommargin{0.5in}
\setleftmargin{0.568in} \setrightmargin{0.768in}
\usepackage{graphicx}
\usepackage{indentfirst}
\usepackage{epsfig}
\usepackage{amsfonts}
\usepackage{color}
\usepackage{amsmath}
\usepackage{algorithm}
\usepackage{algorithmic}
\usepackage{subfigure}
\usepackage{amsmath}
\usepackage{caption}
\IEEEoverridecommandlockouts

\title{Exploiting User Mobility for Wireless Content Delivery}
\author{\IEEEauthorblockN{Konstantinos Poularakis and Leandros Tassiulas}
\IEEEauthorblockA{Department of Electrical and Computer Engineering,\\
     University of Thessaly, Volos, Greece\\
     kopoular,leandros @inf.uth.gr}
\IEEEcompsocitemizethanks{
\IEEEcompsocthanksitem
This work has been published in IEEE International Symposium on Information Theory (ISIT), 2013, Istanbul, Turkey.
}
}

\begin{document}

\maketitle
\pagestyle{empty}
\pagenumbering{arabic}

\newtheorem{property}{Property}
\newtheorem{proposition}{Proposition}
\newcommand{\be}{\begin{itemize}} \newcommand{\ee}{\end{itemize}}
\newcommand{\tb}{\textbf} \newcommand{\ttt}{\texttt}
\newcommand{\tit}{\textit} \newcommand{\uline}{\underline}
\newcommand{\argmin}{\operatornamewithlimits{argmin}}
\newcommand{\argmax}{\operatornamewithlimits{argmax}}
\newtheorem{theorem}{Theorem} \newtheorem{lemma}{\bf Lemma}
\begin{abstract}
We consider the problem of storing segments of encoded versions of content files in a set of base stations located in a communication cell. These base stations work in conjunction with the main base station of the cell. Users move randomly across the space based on a discrete-time Markov chain model. At each time slot each user accesses a single base station based on it's current position and it can download only a part of the content stored in it, depending on the time slot duration.
We assume that file requests must be satisfied within a given time deadline in order to be successful.
If the amount of the downloaded (encoded) data by the accessed base stations when the time deadline expires does not suffice to recover the requested file, the main base station of the cell serves the request.
Our aim is to find the storage allocation that minimizes the probability of using the main base station for file delivery. This problem is intractable in general. However, we show that the optimal solution of the problem can be efficiently attained in case that the time deadline is small. To tackle the general case, we propose a distributed approximation algorithm based on large deviation inequalities. Systematic experiments on a real world data set demonstrate the effectiveness of our proposed algorithms.
\end{abstract}

\section{Introduction} \label{section:1}
Today we are witnessing an unprecedented worldwide growth of mobile data traffic that is expected to continue at an annual rate of 78 percent over the next years, reaching 10.8 exabytes/month by 2016 [10].
These developments lead the research community to investigate ways for increasing area spectral efficiency of cellular networks. The most promising way to achieve this is the deployment of small base stations near to the users that handle a fraction of wireless content traffic in place of the conventional base station of the cell. The drawback of this approach is the high price of the required backhaul to the main base station. The recent work in [1] showed that equipping these base stations with storage capabilities alleviates the backhaul cost. However, little work has been done on exploiting the user mobility in taking storage management decisions in these systems.

A user of these systems is connected to the base station that is currently in communication range based on it's geographical position. The dense spatial deployment of the base stations highlights the scenario that mobile users connect to more than one base stations as they move over time. This transition can happen in a \emph{short time period} of some minutes taking into account the typical values of cell radius $(\approx 400m)$ and the deployment of some decades of base stations in the cell. Thus, only a \emph{fraction} of the requested data may be downloaded by the connected base station.
Besides, depending on the user preferences, it may be acceptable for it to wait a time period until the requested file is delivered. Thus, the user can fetch the entire content file that wishes from different base stations that encounters as it moves within a given time deadline.

In this work, we allow coding in storage decisions at the base stations. This means that encoded versions of the content are stored at the caches instead of the raw data packets. By using an appropriate code, successful content delivery occurs when the total amount of the downloaded data by the encountered base stations within the time deadline is at least the size of the requested file [4].
If the amount of the downloaded data when the time deadline expires does not suffice to recover the requested file, the main base station of the cell serves the request.

In this paper, we focus on the above storage allocation problem at the base stations of the cell. Our goal is to minimize the fraction of file requests that are served by the main base station. The total amount of content stored at a base station is upper bounded by the capacity of it's cache. Traditionally, storage allocation at a cache node is performed based on the average content popularity near that node. In contrast, we \emph{exploit the user mobility in storage decision taking}.
We assume that the user mobility pattern follows a discrete-time Markov model. This is a realistic assumption, as the future position of a user highly depends on it's current position. Thus, the user movement within a given time deadline can be represented by a random walk on a Markov chain. The probabilities of transition can be efficiently derived using learning mechanisms [7].

Our work builds upon the \emph{Femtocaching architecture} proposed in [1] by exploiting user mobility statistics for storage allocation decision taking. The technical contributions of this paper can be summarized as follows:
\begin{itemize}
\item Specifying the coded storage allocation problem.
\item Presenting an \emph{optimal} solution of the problem for short time deadline based on branch and bound algorithms. 
\item Proposing a \emph{distributed} approximation algorithm for arbitrarily large time deadline: We minimize an appropriate probability bound.
\item Evaluating the proposed schemes: We use \emph{traces} of a data set of real mobile users.
\end{itemize}

The remainder of the paper is organized as follows: Section II presents the related work and Sec. III presents the system model and the problem formulation. In Sec. IV we show that the problem is tractable for short time deadline case. Sec. V develops a distributed approximation algorithm by minimizing an appropriate probability bound. In Sec. VI we present our evaluation results. A summary concludes the paper in Sec. VII.
\section{Related work} \label{section:2}
The problem of storing segments of encoded versions of content in a distributed storage system is a well studied one in literature. Ntranos et al. [2] studied the above problem aiming to maximize the probability of recovery of the content after a random set of nodes fail given a total storage budget.
Their work generalizes the results in [4], where homogeneous reliability parameters were assumed.
Furthermore, recent works in [3] and [5] studied the above problem in a delay tolerant network setting assuming that data recovery must be achieved within a given time deadline.

However, the methodologies used in all these works depend on the assumption of independent access of the data stored at the nodes. Thus, they are suitable for representing node failures scenarios or node encounters in delay-tolerant networks. In contrast, our work assumes that data access follows a Markov chain random model, that naturally represents the user movement in a cellular network. Besides, most of the existing work in the area simply assumes that as long as a user contacts a storage node, the complete requested data can be downloaded [3]. In contrast, we consider the realistic case that contact duration limits bottleneck the data transmission.
Finally, our work builds upon a novel network architecture [1] by exploiting user mobility statistics for strategic data placement.

\section{System model and problem formulation} \label{section:3}
We consider a single cell in which the main base station (MBS) serves the content requests of the mobile users that lie in it. A set of $n$ smaller base stations are geographically deployed in the cell. Following the notation in [1] we name these base stations as $helpers$. Each helper $h$ is endowed with a cache of size $|C_h|$. Let $O$ denote a static collection of $|O|$ content files of sizes $|O_i|$, $i=1,...,|O|$. We consider the case that the coverage areas of the helpers are non-overlapping. Thus, a user is in communication range with the nearest helper each time. We denote by $P_{i/h}$ the probability that a user generated request that happens in the area around helper $h$, corresponds to file $O_i$.

As the users move in space they encounter different helpers over time.
In order to represent the system evolution we use a time-homogeneous discrete-time Markov chain $M$ of $n$ states named as $X_1,...,X_n$, where state $X_h$ denotes that a specific user accesses helper $h$. The initial probability distribution of the chain is denoted by $P_{init}$. Clearly, a large value of $P_{init}(h)$ means that the area around helper $h$ is highly populated. The probability of transition $M_{h',h}$, denotes the probability that a specific user encounters the helper $h$ at a time slot, given that the previous time slot the user was connected to the helper $h'$. Typically, pairs of helpers that are neighbors will present higher probabilities of transition than the remote pairs.
The time slot duration is limited, resulting that at most $b_h$ bytes of data can be downloaded each time slot a user contacts helper $h$. The different values of the parameters $b_h$ reflect the bandwidth and the average workload heterogeneity of the helpers [1].

Let $v=(V_1,V_2,...,V_d)$ denote a $d-$step random walk on $M$, where $V_i \in \{1,...,n\}$. Observe that a user can encounter more than one times the same helper within the time deadline $d$. Thus, $v$ is a \emph{multiset}. The probability of accessing the set of helpers in $v$ by a mobile user within $d$ equals to:
\begin{align*}
P(v)=    p_{init} (V_1) \prod_{   h=1    }^{d-1}   M_{V_h,V_h+1} \text{ }\text{ }\text{ }\text{ }\text{ }\text{ }\text{ }\text{ }\text{ }\text{ }\text{ }\text{ } \text{ }\text{ }\text{ }\text{ }\text{ }\text{ }    (1)
\end{align*}

Figure 1 depicts the discussed system model.
\begin{center}
\includegraphics[scale=0.4]{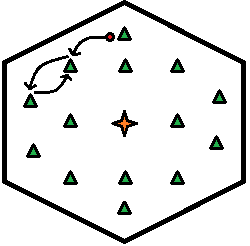}
\end{center}
\begin{flushleft}
{\small
Figure 1. Graphical illustration of the discussed model. MBS lies on the center of the cell. The helpers (triangles) are deployed around MBS. A mobile user (circle) connects to a subset of the helpers within a given time deadline $d$. Here $d=4$ and $v$ consists of three distinct helpers, one of which appears twice in $v$.}
\end{flushleft}

In this work, we focus on determining the storage allocation of these $|O|$ files at the $n$ helpers aiming to minimize the fraction of file requests that are served by the main base station of the cell. Let the optimization variable $x_{h,i}$ indicate the fraction of (encoded) data of file $O_i$ stored at helper $h$ to $|O_i|$. We denote with $\mathcal{M}_{n,d}$ the set of all possible multisets consisted of elements in $\{1,2,...,n\}$ of size $d$. We also denote with $\mathcal{S}_{v}$ the set of distinct helpers that comes from the multiset $v$ by removing the duplicate elements. Let $\eta^h_v$ denote the number of times that helper $h$ appears in $v$.
Then, the probability of failed file delivery corresponding to an allocation $x$ equals:
\begin{align*}
P_f (x) =    \sum_{v \in \mathcal{M}_{n,d} }    P(v)  \sum_{i \in O} P_{i/V_1}  I_{\{  \sum_{h \in \mathcal{S}_{v}} \sum_{k=1}^{\eta^h_v}u_{h,i}^k < 1 \}}(2)
\end{align*}
where $I_{\{.\}}$ is the indicator function, i.e. $I_{\{S\}}=1$ iff $S=true$, else $I_{\{S\}}=0$.
$u^k_{h,i}$ denotes the fraction of file $i$ that can be downloaded from helper $h$ when a user contacts $h$ for the $k^{th}$ time. Clearly, \emph{there is no benefit to download again the same data already downloaded from $h$ at previous contacts}. Thus,
$u_{h,i}^1=\min{\{x_{h,i},\frac{b_h}{|O_i|} \}}$ and $u^k_{h,i}=\min{\{x_{h,i}-\sum^{k-1}_{l=1} u^{l}_{h,i}, \frac{b_h}{|O_i|} \}}$, $k=2,3,...,d$.

The problem of performing the storage allocation that minimizes the probability of failed file delivery is the following:
\begin{align*}
\min _{x}  \text{ }& P_f(x) &(3)\\
s.t.\text{ } &  \sum_{i=1}^{|O|} |O_i| x_{h,i} \leq |C_h|, \forall h=1,...,n&(4)\\
& x_{h,i} \in [0,1], \forall h=1,...,n, i=1,...,|O| &(5)
\end{align*}
,where inequalities in (4) denote the cache capacity constraints.
Inequalities in (5) indicate the non-negativeness of the optimization variables and that it is wasteful to allocate to a cache more than one unit of the same file. The above problem is difficult to solve due to it's non-convex nature and the high number of different multisets of helpers that a user can encounter within the time deadline $d$. Clearly, there exist $n^d$ such multisets.

\section{Small-scale optimal solution } \label{section:4}
In this section we show how to solve optimally the discussed storage allocation problem in a small scale. In other words, we focus on scenarios consisting of a small number of helpers (some decades) and a small deadline parameter $d$. Because of the small value of $n^d$, we can implicitly enumerate all the possible multisets of helpers encountered by a user within $d$.

We first define the binary variable $T_i^v$ to denote whether a user can download sufficient amount of data of file $i$ from the helpers encountered according to the random walk $v$. Thus, $T_i^v$ is defined as follows:
\begin{align*}
T_i^v=
\begin{cases}
 1,  & if \text{ }   \sum_{h \in \mathcal{S}_{v}} \sum_{k=1}^{\eta^h_v}u_{h,i}^k \geq 1 \\
 0,  & else
\end{cases}
\text{ }\text{ }\text{ }\text{ }\text{ }\text{ }\text{ }\text{ }\text{ } (6)
\end{align*}

Then, we can formulate the discussed problem as a \emph{Mixed Integer Programming (MIP)} problem as below:
\begin{align*}
\min _{x,u,T}  \text{ }&   \sum_{v \in \mathcal{M}_{n,d} }    P(v)  \sum_{i \in O} P_{i/V_1}   (1-T_i^v)    &(7)\\
s.t.\text{ } & (4)-(5),\\
& u^k_{h,i} \in [0,\frac{b_h}{|O_i|}], \text{ }k=1,...,d, \text{ } \forall h,i &(8)\\
& \sum_{k=1}^{d} u_{h,i}^k \leq x_{h,i}, \text{ }\forall h,i &(9)\\
& \sum_{h \in \mathcal{S}_{v}} \sum_{k=1}^{\eta^h_v}u_{h,i}^k \geq T_i^v, \forall i,v &(10)\\
& T_i^v \in \{0,1\}, \forall i,v &(11)
\end{align*}
Inequalities (8)-(9) are added because of the definition of variables $u$. The number of optimization variables is $|O|(n^d+nd+n)$. It is known that the above problem can be efficiently solved using Branch and Bound algorithms [5] for a small number of optimization variables.

\section{Large-scale approximate solution} \label{section:5}
In this section we establish an approximate solution of the discussed problem for arbitrarily large problem instances by minimizing an appropriate probability bound. A similar approach was used in [2]. We start with the following lemma:

\begin{lemma}(Chernoff-Hoeffding probability bound [6])

Let $E$ be an ergodic Markov chain with state space $S$ and stationary distribution $\pi$. Let $(V_1, . . . , V_t)$ denote a $t-$step random walk on $E$ starting from an initial distribution $\phi$ on $S$. For every $i \in \{1,2,...,t\}$, let $f_i : S \rightarrow [0,1]$ be a weight function at step $i$ such that the expected weight $E_v[f_i(v)]= \mu$ for all $i$. Define the total weight of the random walk $(V_1,...,V_t)$ by $X_t= \sum_{i=1}^t f_i(V_i)$. There exists some constant $c$ (which is independent of $\mu$ and $\delta$), such that:
$Pr[X_t \leq (1 - \delta) \mu t ] \leq  c ||\phi||_{\pi} exp(- \frac{ \delta ^2 \mu t } {72 T })$, for $0 \leq \delta \leq 1$, where $T$ is the mixing time of $E$, defined as $T=\min \{  t: \max_q ||qE^t - \pi||_{TV} \leq \frac{1}{8}\}$, $q$ is an arbitrary initial distribution over $E$, $||u-v||_{TV} = \max_{A \subseteq V}|\sum_{i \in A}u_i -\sum_{i \in A}w_i|$ and $||u||_{\pi}=\sqrt{\sum_{x \in S} \frac{u_i^2}{\pi(i)}  }$.
\end{lemma}

Let $Y$ be the random variable indicating the fraction of the requested file that can be downloaded by a user within $d$, given the storage allocation $x$. Then, $P_f(x)=P[Y<1]$.
We use lemma 1 to derive an upper bound on the probability $P[Y < 1] \leq P[Y \leq 1]$ for an arbitrary allocation. $Y$ can be interpreted as the total weight $\sum_{i=1}^d f_i(V_i)$ of the random walk $(V_1,...,V_d)$, on an appropriately constructed markov chain $E$.
We construct the Markov chain $E$ as follows:

The state space $S$ of $E$ consists of $1+|O|\sum_{i=1}^{d}(n^i)$ states. We name one of the states of $S$ as the \emph{root}, indexed by $0$. We partition the other states into $|O|$ groups, such that the first group contains the first $\sum_{i=i}^d (n^i)$ states, the second group contains the next $\sum_{i=i}^d (n^i)$ states etc.
The states of these groups combined with the root form an \emph{hierarchy} of $d+1$ levels named as $\{0,1,...,d\}$. Root is the only state of level 0. Root is the unique parent of the $|O|*n$ states of level $1$, $n$ states for each of the groups.
Each state that is included in the level $l\in \{1,...,d-1\}$ of the group $g \in \{1,...,|O|\}$ is the unique parent for $n$ states of the level $l+1$ of that group. We define by $child(u,c)$ the $c^{th}$ child-state of state $u$. Besides, let $\mathcal{G}(u)$ denote the group that contains state $u$ and $\mathcal{P}_{u}$ denote the set of states that lie on the path between state $0$ and state $u$ (including the two endpoints). A state $u$ that belongs to the level $l>1$ and it is the $c^{th}$ child of another state, represents the connection of a user requesting file $\mathcal{G}(u)$, to helper $c$ at the time slot $l$. The $n$ first states of level 1 represent the connection to a node requesting file $1$ at time slot 1, the next $n$ states represent the same but requesting file $2$ etc. Figure 2 illustrates a simple example of $E$.

The initial probability $\phi $ of $E$ is given by:
\begin{align*}
\phi (u)=
\begin{cases}
P_{\mathcal{G}(u)/c'} P_{init}(c'),  & if \text{ }  {\exists c \in \{1,...,n*|O|\}:\atop u=child(0,c)\text{ }\text{ }\text{ }\text{ }\text{ }\text{ }  }\\
0, & otherwise
\end{cases}
\end{align*}
where, for ease of presentation, we denoted by $c'=c-(\mathcal{G}(u)-1)n$. According to $\phi()$, every walk starts at a state of level $1$.

The probability of transition from state $v$ to $u$ is given by:
\begin{align*}
E_{v,u}=
\begin{cases}
1-\alpha , &  if \text{ } v=u=0\\
\alpha P_{\mathcal{G}(u)/c'} P_{init}(c') , &  if \text{ }  {v=0, \text{ }\text{ }\text{ }\text{ }\text{ }\text{ }\text{ }\text{ }\text{ }\text{ }\text{ }\text{ }\text{ }\text{ }\text{ }\text{ }\text{ }\text{ }\text{ }\text{ }\text{ }\text{ }\text{ }\text{ }\text{ }\text{ }\text{ }\text{ }\text{ }\text{ }\text{ }\text{ }\text{ }\text{ }\text{ } \atop \exists c \in \{1,...,n|O|\} : u=child(0,c)} \\
M_{c',e},  & if \text{ } {\exists c \in \{1,...,n|O|\}, e \in \{1,...,n\}:\atop v=child(0,c), u=child(v,e)\text{ }}  \\
M_{c,e},  & if \text{ } {\exists c,e \in \{1,...,n\}, w \in S \setminus 0:  \text{ }\text{ }\text{ }\text{ }\text{ }\text{ }\text{ }\text{ }\text{ }\text{ }\atop v=child(w,c), u=child(v,e)\text{ }}  \\
1  , & if \text{ }{  |\mathcal{P}_{v}|=d+1, u=0       }\\
0, &otherwise
\end{cases}
\end{align*}
where $a \in (0,1)$ is a fixed parameter. The transitions from the states of the $d^{th}$ level to state 0 and from state 0 to itself are necessary for $E$ to be ergodic.
The weight of a state $u$ is defined as:
\begin{align*}
f_i(u)=
\begin{cases}
u_{ c', \mathcal{G}(u) }^1, &if \text{ } \exists c \in \{1,...,n|O|\}:  u=child(0,c) \\
u_{ c, \mathcal{G}(u) }^1, &if \text{ }  {\exists c \in \{1,...,n\}, \text w \in S \setminus 0: \text{ }u=child(w,c),\atop \not \exists y \in \mathcal{P}_{u} \setminus u:  \text{ }f_i(y)=u^{1}_{c,\mathcal{G}(u)} \text{ }\text{ }\text{ }\text{ }\text{ }\text{ }}\\
u_{ c, \mathcal{G}(u) }^k, &if \text{ }  {\exists c \in \{1,...,n\}, \text w \in S \setminus 0, \text{ }k \in Z^+: \text{ }u=child(w,c),\atop k=max \{ k: \exists y \in \mathcal{P}_{u} \setminus u:  \text{ }f_i(y)=u^{k-1}_{c,\mathcal{G}(u)} \} }\\
0,   &else
\end{cases}
\end{align*}

Observe that according to the definition of the $f_i()$ function, the weight of a state $u$ matches the fraction of the requested file that a user can download during the associated contact. This weight depends only on the number of the previous contacts of the user to the same helper, as explained in section III.
\begin{figure}
\centering
\includegraphics[scale=0.2157]{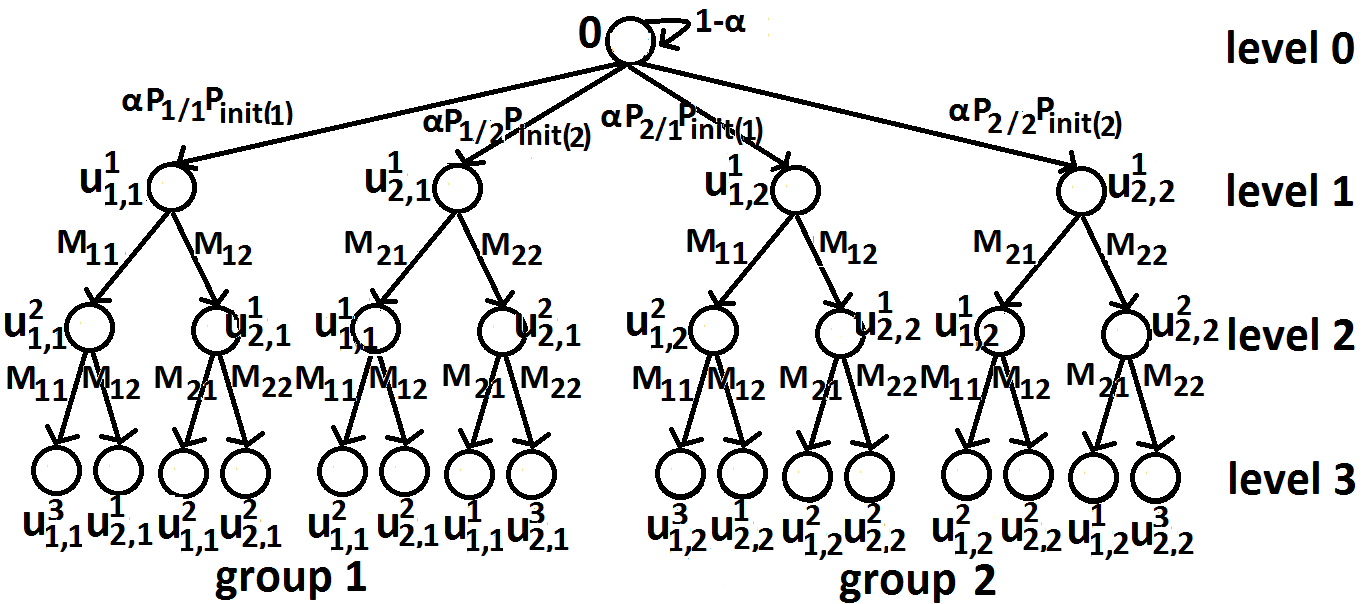}    
\vspace{-2mm}
\begin{flushleft}
{\small
Figure 2. An example of the Markov chain $E$. The parameters are set as follows: $n=2$, $d=3$ and $|O|=2$. Circles represent the states and arrows the transitions of non zero probability. State labels represent the state weights. Arrow labels represent the probabilities of transition. For ease of presentation we omitted the transitions that start from states of level 3. In reality, for each of these states there is a transition to the root state of probability $1$.}
\end{flushleft}
\vspace{-8mm}
\end{figure}

By construction of $E$, every $d-$step walk starts at a state of level 1 and ends at a state of level $d$, representing a walk of a user over the helpers within $d$. The requested file is specified by the group that includes these states. $M_{h',h}$ and $P_{i/h}$ may be equal to zero for some values of $i$, $h$ and $h'$. In order to ensure ergodicity of $E$, we exclude the states that are reachable by the root state with zero probability.
Let $\mu(u,x)$ be the expected weight of a walk on $E$ as a function of the weights. Set $\delta=1-\frac{1}{\mu(u,x) d}$ and require $\mu(u,x) d \geq 1$.  Then, lemma 1 yields:
\begin{center}
$P[Y<1] \leq c ||\phi||_{\pi} exp(- \frac{ \mu(u,x) d + \frac{1}{\mu(u,x) d} -2 } {72T })\text{ } (12)$
\end{center}
,where $\pi$ is the stationary distribution of $E$, $T$ is the mixing time of $E$ and $c$ is a constant. By definition, it holds that:
\begin{center}
$\mu (u,x) = \sum_{h=1}^{n} \sum_{k=1}^{d} \sum_{i=1}^{|O|}  P[u^k_{h,i}] u^k_{h,i} \text{ }\text{ }(13)$
\end{center}
where $P[u^k_{h,i}]$ denotes the probability of reaching any of the states of $E$ with weight equal to $u^k_{h,i}$ according to $\pi$.

Note that $\pi$ can be easily computed because of the special form of $E$; for any state $u \in S \setminus 0$ it is $\pi(u)=\pi(0)*E_{0,i_1}*E_{i_1,i_2}*...*E_{i_l,u}$, where $\mathcal{P}_{u}=\{0,i_1,...,i_l,u\}$. Besides, it holds that $\sum_{u \in S}\pi(u)=1$. Thus, after some computations we find that $\pi(0)=\frac{1}{1+\alpha d}$.
$P[u^k_{h,i}]$ can be defined as follows:

\begin{align*}
&\frac{P[u^k_{h,i}]}{\pi (0)\alpha} =  P_{init}(h)P_{i/h}*\sum_{{l_1+..+l_{k-1} \leq d-1 \atop l_1,..,l_{k-1} \geq 1  }} (\prod^{k-1}_{i=1}r_{h,h}(l_i)) +\\
&+\sum_{h' \neq h}(P_{init}(h')P_{i/h'} *\sum_{{l_1+..+l_{k} \leq d-1 \atop l_1,..,l_{k} \geq 1  }} r_{h',h}(l_1)\prod^{k}_{i=2}r_{h,h}(l_i)) (14)
\end{align*}
,where $r_{i,j}(l)$ denotes the probability that a user connects to the helper $j$ \emph{for the first time} at the $l^{th}$ time slot conditioned on the event that it connects to the helper $i$ at the first time slot. $r_{i,j}(l)$ is defined by the following set of \emph{recursive equations}:
\begin{center}
$r_{i,j}(l)=\sum_{k=1, k \neq j}^{n} r_{i,k}(l-1)M_{k,j}$\text{ }\text{ } \text{ }\text{ }\text{ }(15)\\
$r_{i,j}(1)=  M_{i,j}\text{ }\text{ }\text{ }\text{ }\text{ }\text{ }\text{ }\text{ }\text{ }\text{ }\text{ }\text{ }\text{ }\text{ }\text{ }\text{ }\text{ }\text{ }\text{ }\text{ }\text{ }\text{ }\text{ }\text{ }\text{ }\text{ }\text{ }\text{ }\text{ }\text{ }\text{ }\text{(16)}$
\end{center}

Equations (15)-(16) are very similar to the Chapman-Kolmogorov equations. However, Chapman-Kolmogorov equations do not require $k \neq j$ over the above summation. Simply speeking, equation (14) specifies that a user can initially be at helper $h$ and then encounters $h$ again $k-1$ times, or the user can initially be at a different helper $h'$ and then encounters $h$ $k$ times before the deadline $d$ expires.

In order to derive an approximate solution of the problem described in (3)-(5) we minimize the upper bound in (12). Lemma 2 shows that minimizing the aforementioned bound is equivalent to maximizing the expected weight $\mu(u,x)$.

\begin{lemma} Let $g(u,x)= c ||\phi||_{\pi}\exp(-\frac{\mu(u,x) d + \frac{1}{\mu(u,x) d} -2 }{72T}) $.
\begin{center}
$\text{Then,}\text{ }\underset{(u,x) \in \mathcal{A}}{\argmin} \text{ } g(u,x)  =
 \underset{(u,x) \in \mathcal{A}}{\argmax} \text{ } \mu (u,x)  \text{ }\text{ }\text{ }\text{ }\text{ }\text{ }\text{ }\text{ }\text{ }\text{ }\text{ }\text{ }\text{ }\text{ }\text{ }\text{ }\text{ }$
\end{center}
where $\mathcal{A}=\{ (u,x) \in \mathbb{R}^{n \times |O| \times d} \times \mathbb{R}^{n \times |O|} : (4)-(5),(8)-(9) \text{ are satisfied}\}$.
\end{lemma}
\begin{proof}
Let $(u^*,x^*) = \underset{(u,x) \in \mathcal{A}}{\argmin} \text{ } g(u,x)$. Then, $g(u^*,x^*) \leq g(u,x)$, $\forall (u,x) \in \mathcal{A} $. Dividing by $c||\phi||_{\pi}$ and then taking the logarithm on both sides preserves the inequality as $c > 0$, $||\phi||_{\pi} > 0$ and $log(x)$ is strictly increasing. Thus, we have:
\begin{center}
$-\frac{\mu(u^*,x^*) d + \frac{1}{\mu(u^*,x^*) d} -2 }{72T} \leq -\frac{\mu(u,x) d + \frac{1}{\mu(u,x) d} -2 }{72T}$
\end{center}
Dividing by $-72T \leq 0$ and then adding 2 on both sides yields:
\begin{center}
$\mu(u^*,x^*) d + \frac{1}{\mu(u^*,x^*) d}  \geq  \mu(u,x) d + \frac{1}{\mu(u,x) d}$
\end{center}
However, it holds that: $\mu(u,x) d \geq 1$, resulting that:
$\mu(u^*,x^*) d \geq \mu(u,x) d, \text{ } \forall (u,x) \in \mathcal{A}$ \end{proof}

The optimization problem becomes as follows:
\begin{align*}
\max _{u,x}  \text{ }& \mu(u,x) &(17)\\
s.t. \text{ }&(4),(5),(8),(9)&
\end{align*}

By the structure of $\mu(u,x)$, that is defined in (13), and the above problem constraints, we observe that the storage allocation decisions at a helper do not affect the storage allocation decisions at the other helpers. Thus, we can \emph{decompose} the problem to $n$ independent \emph{linear programming} subproblems, one for each helper, and solve them in a \emph{distributed} manner.
It is known that he numerical solution of a linear programming problem can be efficiently attained using the simplex method. However, as we show below each of these subproblems falls into a class of tractable problems, with \emph{known solution structure}, alleviating the need for applying the simplex method.

Fractional knapsack problem asks for placing fractions of materials of different values and weights in a knapsack of limited capacity in order to maximize the aggregate value of the materials placed in it.
The storage allocation at a helper $h$ subproblem can be translated to a \emph{restricted version of the fractional knapsack problem} in which there exist $|O|*d$ materials, one material for each variable $u^k_{h,i}$, $i=1,...,|O|$, $k=1,...,d$ and a knapsack of capacity $|C_h|$. The value of the material corresponding to $u^k_{h,i}$ is $P[u^k_{h,i}]$ and it's weight is $|O_i|$.
The material placement must also satisfy the following two constraints:
1) We are restricted to place in the knapsack at most a $\frac{b_h}{|O_i|}$ fraction of the material corresponding to variable $u^k_{h,i}$.
2) The sum of the fractions of the materials corresponding to variables $u^k_{h,i}$, $k=1,...,d$, placed in the knapsack must not be greater than 1.
Using the exact same arguments used in [8], we can prove that the optimal solution of this knapsack-type problem can be attained by the following \emph{greedy} algorithm:

\fbox{ \parbox{0.9\linewidth} {Sort the materials in decreasing order of value per unit of weight. Then insert them into the knapsack, starting with as
large amount as possible of the first material, without violating any of the above two constraints, until there is no longer space in the knapsack for more.}}
\vspace{+1mm}

The solution of this knapsack-type problem, attained by the greedy algorithm, translates to a solution to the original storage allocation at helper $h$ subproblem such that $u^k_{h,i}$ variable takes the value equal to the fraction of the associated material placed in the knapsack and $x_{h,i}$ takes the value equal to $\sum_{k=1}^d u^k_{h,i}$.
The solution is independent of the values of $\alpha$ and $\pi$.
The complexity of the greedy algorithm comes mainly to the sorting of the values of the materials. Quicksort is the fastest sorting algorithm, resulting \emph{complexity} of $|O|d \log (|O|d)$.

\section{Performance evaluation} \label{section:6}
In this section we present the numerical experiments that we have conducted to evaluate the performance of the proposed algorithms. The algorithms have been applied to a cellular network consisting of $n=623$ helpers. The geographical position of the helpers as well as the mobility pattern of the users were acquired by the recording of wireless traffic logs included in the CRAWDAD data set [9]. The set contained 13888 logs, one for each mobile user, including details about the encountered helpers as well as the time that each encounter happened.
According to our model, location transitions of the users happen in a time slotted fashion. We set the time slot duration to be equal to $100$ seconds. Based on the data set, we set $P_{init}(h)$ to be equal to the frequency of time slots at which a user starts it's walk from helper $h$. Similarly, $M_{i,j}$ takes as value the frequency of time slots at which a user encounters sequentially within the same time slot the helpers $i$ and $j$. If the user does not encounter any other helpers by the end of the time slot, then we assume that it encounters again the same helper increasing the value of $M_{i,i}$.

In all simulations, we assume a collection of $|O|=100$ files each one of size $30$ MB. We use a Zipf-Mandelbrot model to formulate the files request pattern with a shape parameter $alpha$ and a shift parameter $q=10$. At each time slot at most $b_h = 15$ MB of the requested data file can be downloaded by the encountered helper $h$. A user is satisfied when the requested file is delivered to it within a time deadline that is equal to $300$ seconds, i.e. $d=3$ time slots. Such a time deadline is reasonable for a playback time for a typical video file.

Throughout, we compare the probability of failed file delivery achieved by three algorithms:
\begin{enumerate}
\item Heuristic Uncoded Algorithm $(HUA)$: the standard mode of operation currently in use in most storage systems. Each helper stores the most popular files that fit in it's cache independently from the others.
\item Approximation Coded Algorithm $(ACA)$: The solution of the knapsack-type problem described in section V.
\item Optimal Coded Algorithm $(OCA)$: The solution of the MIP problem described in section IV.
\end{enumerate}

\begin{figure}[ht]
\centering
\subfigure{
    \includegraphics[scale=0.295]{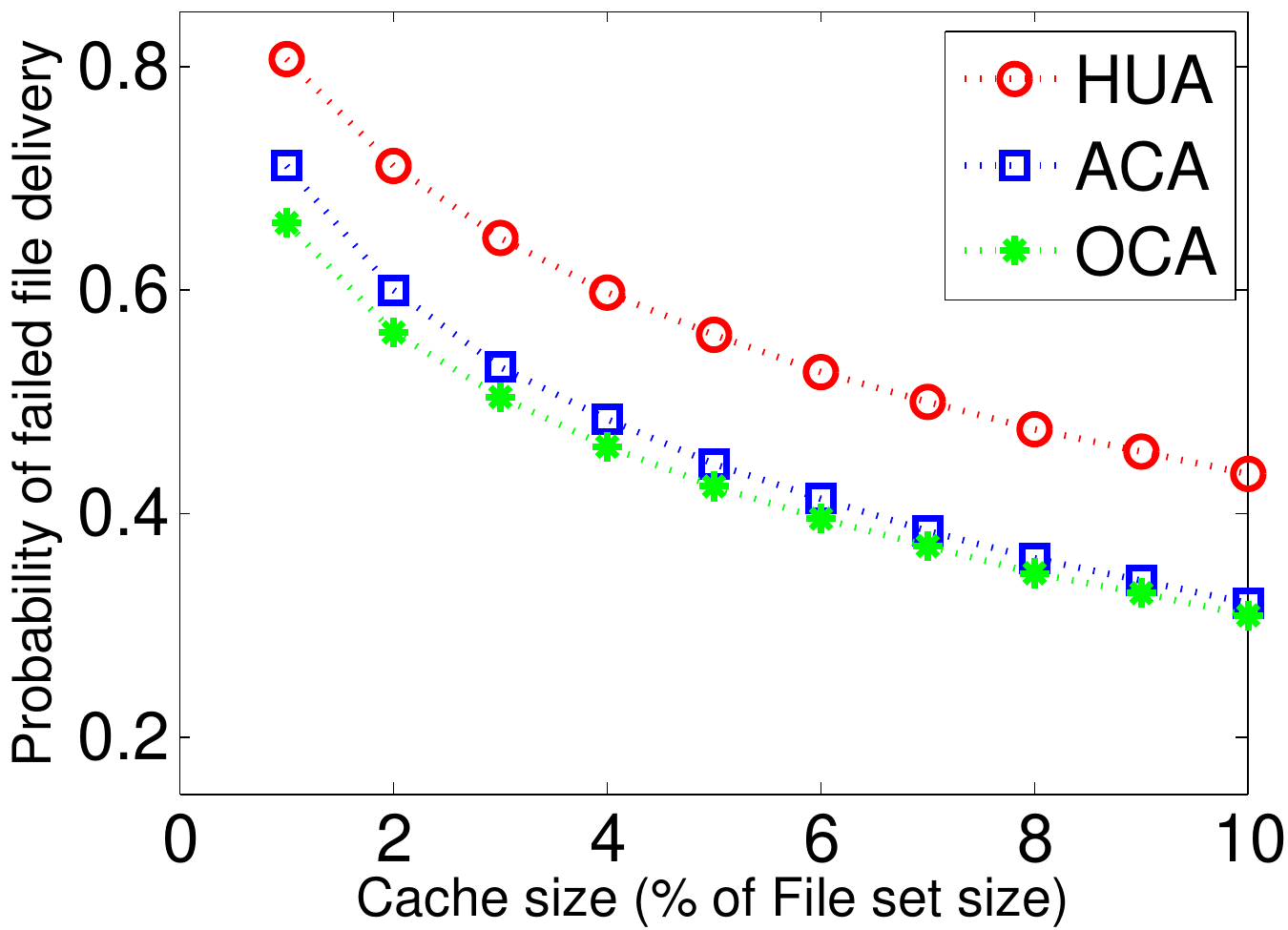}
}
\subfigure{
    \includegraphics[scale=0.2805]{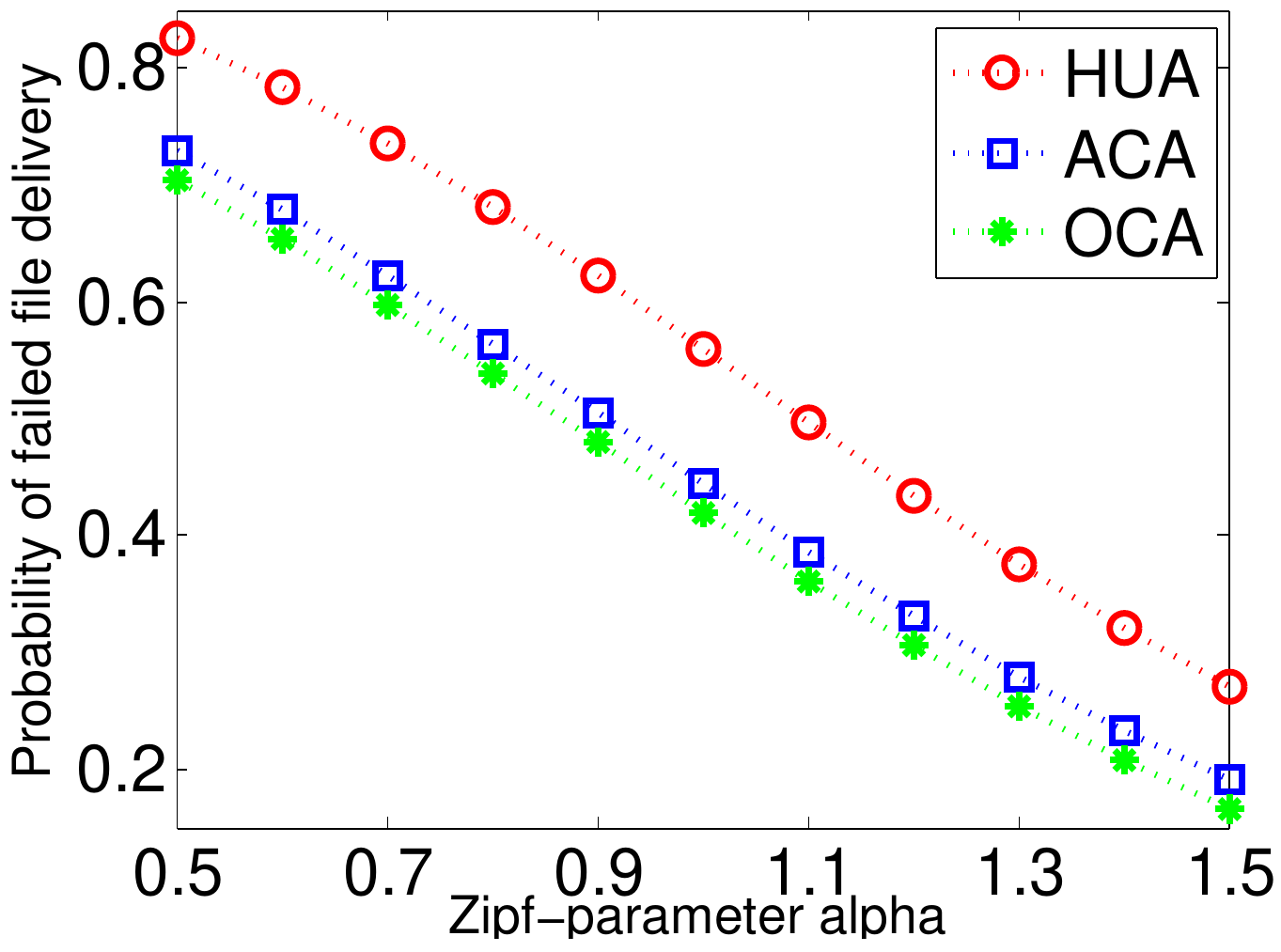}
}
\vspace{-1mm}
\text{ } \text{ } a) Impact of cache sizes   \text{ } \text{ }  b) Impact of Zipf-parameter
\vspace{-1mm}
\begin{flushleft}
{\small
Figure 3. Probability of failed file delivery as a function of a) the cache sizes and b) the Zipf-parameter $alpha$.}
\end{flushleft}
\vspace{-7mm}
\end{figure}

Figure 3(a) shows the results as a function of the cache size of the helpers for $alpha=1$. The cache size of each helper is varied from $1\%$ to $10\%$ of the entire file set size. As expected, increasing the cache size of each helper reduces the probability of
failed file delivery.
Figure 3(b) shows the results as a function of the parameter $alpha$. The cache size of each helper was set to
$5\%$ of the entire file set size. We can see from the graph that the probability of failed file delivery decreases with increasing values
of $alpha$, reflecting the well known fact that caching effectiveness improves as the popularity distribution gets steeper.
In general, we observe that $OCA$ is strictly better than $ACA$, which in turn is better than the $HUA$.
The performance achieved by the $OCA$ and $ACA$ are very close. This indicates that helpers can independently take
local storage decisions, and still have a significant gain. In summary, our algorithms perform $20-50\%$ better than the
conventional file placement scheme.

\section{Conclusion} \label{section:7}
In this paper we introduced a new storage allocation scheme for offloading traffic from the cellular network.
Our work builds upon a recent network architecture [1], with concerns on user mobility and limited contact duration time. Our main contribution is a distributed light-weight storage allocation algorithm.
We used a real trace of user movements and demonstrated significant performance gains compared to conventional schemes. Our algorithms can be easily extended to handle the case that the coverage areas of the helpers are overlapping.

\end{document}